\documentclass{llncs}
\usepackage{amsfonts,amssymb,amsmath}
\usepackage{graphicx,xspace}

\newcommand{\PNP}{P $\neq$ NP\xspace}
\newcommand{\Past}[1]{\ensuremath{\mathop{\mathit{Past}}({#1})}}
\newcommand{\Abs}[1]{\ensuremath{\mathopen{|}{#1}\mathclose{|}}}
\newcommand{\HowFarBack}[2]{\Abs{ {#1} - {#2} }}
\newcommand{\Ribbon}{\ensuremath{\mathit{Ribbon}}\xspace}

\begin{document}

\title{Computation and Spacetime Structure}
\author{Mike Stannett}
\institute{University of Sheffield\\
  \email{m.stannett@dcs.shef.ac.uk}
}
\date{6 March 2011}
\maketitle

\begin{abstract}
We investigate the relationship between computation and spacetime structure, focussing on the role of closed timelike curves (CTCs) in promoting computational speedup. We note first that CTC traversal can be interpreted in two distinct ways, depending on ones understanding of spacetime. Focussing on one interpretation leads us to develop a toy universe in which no CTC can be traversed more than once, whence no computational speedup is possible. Focussing on the second (and more standard) interpretation leads to the surprising conclusion that CTCs act as perfect information repositories: just as black holes have entropy, so do CTCs. If we also assume that \PNP, we find that all observers agree that, even if unbounded time travel existed in their youth, this capability eventually vanishes as they grow older. Thus the computational assumption \PNP is also an assumption concerning cosmological structure.
\end{abstract}

\section{Introduction}
\label{sec:introduction}

In the presence of spacetime curvature, the run-time of a program typically depends on who does the observing; the time registered by a clock co-moving with a computational system may differ from that registered by an observer watching the system from elsewhere. The existence of such discrepancies lies at the heart of \emph{relativistic hypercomputation} schemes using e.g. Malament-Hogarth spacetimes \cite{Hog04}, slow Kerr black holes \cite{One95,EN02} and closed timelike curves \cite{ANS11}. These schemes indicate that cosmological anomalies allow the resolution of formally undecidable problems, so it seems not unlikely that they would also allow problems in NP$\setminus$P (if any) to be solved in polynomial time. We investigate in this paper whether this is necessarily the case.

Suppose, then, that we live in a universe which contains closed timelike curves (CTCs). An observer who traverses a CTC considers himself to be doing nothing out of the ordinary; he travels forward in time as usual, never exceeding light-speed, but eventually finds himself at a point in spacetime he has already visited previously. In a sense, then, the observer has ``travelled into the past'', but it is important to note that at no time does he violate any physical laws as viewed from his own co-moving frame of reference, nor does he consider himself to be moving ``backwards in time''. He is simply following a path through spacetime that happens to include a loop.

Given the capability of time travel, a simple argument then suggests that P = NP. For suppose $A$ is a deterministic program for solving some problem in NP, and construct the algorithm $A'$ in Fig. \ref{fig:algorithm-A'}.

\begin{figure}
\texttt{\begin{itemize}
\item[0] output v
\item[1] start A and let it run to completion
\item[2] let v be the result generated by A
\item[3] send v back to time 0
\end{itemize}}
\caption{Algorithm $A'$ exploits time travel to solve in constant time the same problem that $A$ solves in super-polynomial time.}
\label{fig:algorithm-A'}
\end{figure}
Although $A'$ may run for some considerable time, the observer always obtains the required output at step $0$. The total runtime of $A'$ may be superpolynomial, since it includes step 1 (running $A$) but nonetheless the \emph{result} is produced at step 0, and in this sense $A'$ can be said to solve the problem in a fixed amount of time. Since $A'$ is deterministic and solves in constant time the same problem as $A$, that problem must be in P, whence (loosely speaking) P = NP. 

Unfortunately, this apparently simple argument is logically incomplete, since it relies on unstated assumptions concerning the nature of CTCs, and these assumptions need not be generally valid. There are two essentially distinct ways in which CTCs might be exploited to implement computational speed-up. In the absence (so far) of experimental data confirming the existence of CTCs and the experiences of observers traversing them, the viability of these two computational schemes depends upon ones philosophical interpretation of relativity theory. In this paper we focus on one of these approaches; nonetheless we briefly discuss the consequences of choosing the other interpretation in Sec. \ref{sec:discussion}.

\section{CTC Computation}
\label{sec:version-one}

Consider the following science-fiction clich{\'e}: a historian wants to make a clandestine visit to Ancient Rome, so he selects a suitable CTC and sets off on his journey. He makes detailed notes of Julius C{\ae}sar's activities, and then returns to exactly the point in time and space from which he originally set off, so that his unauthorised absence cannot be detected. He repeats the same deception several days running. Using the information in his notes, he then writes an important academic paper and becomes famous.

Although this kind of story is familiar from science fiction, it requires a particular interpretation of what it means for a body to move in space and time. For consider what happens when the historian `returns to the present'. At this point in the journey, he occupies exactly the same position in time and space as when he originally set off on his journey -- \emph{but he is not constrained to repeat the same behaviour}, for rather than repeating the journey to Ancient Rome, he chooses instead to write an academic paper. Moreover, since he occupies the same spacetime position at both points on his journey, and his notepad is in his pocket both times, its contents should be the same both times -- but it contains notes when he returns which were not present when he set off.

At first sight this seems to suggest a fundamental logical inconsistency, leading to the conclusion that this kind of CTC exploitation is impossible. But there is in fact no contradiction present, provided we think of spacetime as a surface across which bodies move. The fact that a body can occupy a given position more than once, and be in different states each time, is hardly surprising given this interpretation; it is no different to a racing driver completing several laps of a Grand Prix, and then deciding on the next lap that he needs to take a pit stop so that his tyres can be replaced. He may pass through the same positions on the track several times, but he is not thereby constrained to repeat the same behaviour each time.

From the historian's point of view too there is no contradiction, because we have to ask ourselves \emph{in what sense has he travelled back in time}? Certainly, he cannot have done so relative to his own clock, because he considers himself to be moving always forwards in time at sublight speeds. His judgment \emph{I am in Ancient Rome} must therefore have been made relative to evidence provided by some independent witness (for example, he could ask a local trader what year it is, and whether the person standing in front of them is indeed Julius C{\ae}sar). It is entirely possible, of course, that the witness might observe multiple copies of the historian, but this is not contradictory either, for each copy is in a different state (when asked how old he is, each copy of the historian will give a different answer). From the viewpoint of the witness, the various copies of the historian are distinct objects, and there is no sense in which the historian is observed in different places at the same time.

Thus, neither the witness nor the historian observes anything contradictory as a result of his re-occupying a point in spacetime without being constrained to repeat the same behaviour. By the same argument, a computer can be sent around a CTC several times without its computation being forced to loop, so we could run a program for as long as we like simply by traversing a 5-second CTC as often as we like. If we arrange for the machine to produce some observable output if it halts (e.g. by printing a result) we can always decide after 5 seconds whether the program has halted, and if so what its result is, simply by checking the printer (which need not be on the CTC). Thus, not only does hypercomputation seem to be possible in this senario, but all problems in NP can apparently be solved deterministically in constant time, whence P = NP.

This argument is, however, logically flawed, for although a CTC returns an observer to an earlier point on his worldline, \emph{it does not follow that he can traverse the CTC a second time}, as we now show.

\subsection{Single-traversible CTCs}
\label{sec:single-traversible-ctcs}

In this section we introduce a crossed-ribbon toy spacetime that includes a single inhabited CTC. It is impossible for any massive body in this model to traverse more than one CTC, and no CTC can ever be traversed more than once. Our crossed-ribbon universe is in some respects non-standard (but we will nonetheless argue that it is a \emph{reasonable} model).

Our description of the crossed-ribbon universe is in three stages. First we construct a standard $(1+1)$-dimensional spacetime ($M$); then we describe the crossed-ribbon (\Ribbon); and then we populate \Ribbon with bodies and consider their worldlines. Finally, we argue that $M$ can be populated with bodies in such a way that observations in $M$ and \Ribbon are indistinguishable. Since $M$ gives an inherently acceptable description of spacetime, it follows that \Ribbon is also an observationally reasonable model.

\subsubsection{Construction of $M$.}
Following Andr\'eka and her colleagues \cite{Sze09,AMN07} we assume that spacetime is coordinatized by an ordered Euclidean field $Q$. Our starting point is a ribbon-shaped manifold, $M = \left(-\alpha,\alpha\right) \times Q$, where $\alpha \in Q$ (Fig. \ref{fig:looped-ribbon}). We regard this as a (1+1)-dimensional Minkowskian spacetime of infinite length and width $2\alpha$, with time flowing along the length of the ribbon, and space across its width. If we take $Q = \mathbb{R}$, then $\alpha$ should be finite, but for more general coordinatizations this need not be the case: for example, if $Q$ contains infinitesimals, $\alpha$ could be an infinite value in $Q$.

\subsubsection{Construction of \Ribbon.}
We now imagine wrapping $M$ around in 3-dimensional space so that it self-intersects at right-angles. We identify the overlapping regions, and call the resulting manifold \Ribbon (Fig. \ref{fig:looped-ribbon}).

\begin{figure}
\centering
\includegraphics[scale=0.4]{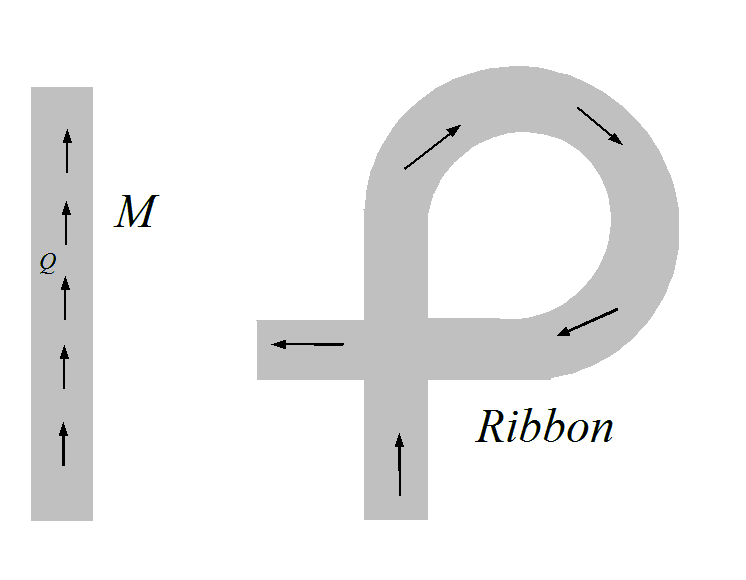}
\caption{The $(1+1)$-dimensional manifold $M$, and the immersed manifold \Ribbon generated by looping $M$ in such a way that it self-intersects at right-angles.}
\label{fig:looped-ribbon}
\end{figure}

\subsubsection{Worldlines in \Ribbon.}
The manifold $M$ is a standard $(1+1)$-dimensional spacetime model, and we can populate it with bodies in the usual way. As usual, we shall assume for convenience that time flows up the page, and define as an allowable worldline any path followed by a body that always travels at subluminal speed. Given any such worldline $w$ in $M$, we define $w'$ to be its image in \Ribbon. We then \emph{reflect} the worldline back into $M$. In other words, we determine what other bodies would need to be present in $M$, and following what worldlines, if a body following $w$ is to observe exactly the same series of events as a body moving along the corresponding path $w'$ in \Ribbon (Fig. \ref{fig:w-and-w'}). 

\begin{figure}
\centering
\includegraphics[scale=0.5]{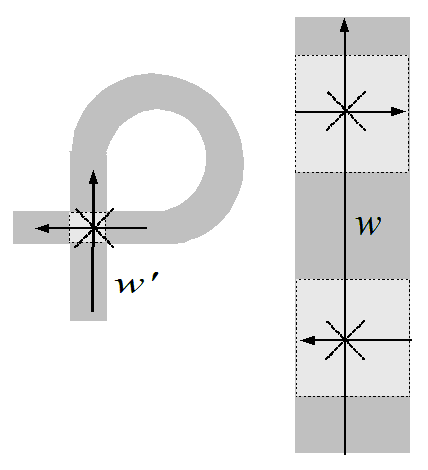}
\caption{Schematic showing how a worldline $w'$ in \Ribbon corresponds to three intersecting worldlines in $M$.}
\label{fig:w-and-w'}
\end{figure}

\subsection{Is \Ribbon a reasonable model of spacetime?}
Clearly, any body following $w'$ will encounter itself twice (once each time it crosses the self-intersection region), but it is important to realise that on each occasion the body considers its other incarnation to be travelling faster than light (FTL), because of the way \Ribbon intersects itself at right angles (which ensures that the time and space axes have been interchanged when the body encounters its past self). Reflecting this back into the original manifold $M$, the body traversing $w$ meets two FTL versions of itself as it moves along its worldline. The paths followed by these FTL bodies are fully determined once $w$ is specified.

Apart from the existence of bodies apparently moving at FTL speeds relative to one another, there is nothing unusual about this three-body version of $M$, and indeed FTL motion has long been a research topic in cosmological theory \cite{MTY88,GD00,Sch10}. We therefore claim that populating $M$ with these additional FTL bodies yields an entirely reasonable (toy) universe. By construction, however, the flow of events observed by the body following $w'$ in \Ribbon is identical to the flow observed by the body following $w$ in the three-body variant of $M$, whence \Ribbon must also be a reasonable (toy) universe.

\subsection{Existence of single-traversable CTCs}
The significance of \Ribbon lies in the nature of the path $w'$. Since this path includes a self-intersection, it implements a CTC. But as we have noted above, when the body meets its former self, its time and space axes have been interchanged, and it considers its past self to be travelling at FTL speeds. In order to re-traverse the CTC it would need to complete the `90 degree turn' at the point of intersection, and this would violate our presupposition that a body following a CTC always considers itself to be obeying physical laws and travelling slower than light.

Thus \Ribbon provides an example of a toy spacetime containing just one inhabited CTC (since \Ribbon only contains one body), where this CTC \emph{cannot} be traversed a second time.\footnote{We could, of course, add more bodies to the model if we wished, but doing so would add nothing to our argument.} As explained above, this undermines the argument that the existence of CTCs is sufficient, in itself, to ensure the existence of computational speed-up and hypercomputation. If we wish to deduce that CTCs move all problems (if any) in NP$\setminus$P into P, we need to impose the \emph{additional} condition that at least one CTC is multiply-traversable.

\section{Discussion}
\label{sec:discussion}

We have described an interpretation of CTC traversal in which spacetime exists as an independent entity across which bodies move subject to various laws; a body could potentially pass through a given location several times without being constrained to display identical behaviour subsequently. An alternative viewpoint is that a point in spacetime is fully defined by the set of bodies that exist there \cite{AMNS11}. According to this viewpoint, when the historian returns to the present, he has no choice but to re-traverse the CTC back to Ancient Rome. Since he occupies the same location in spacetime as his past self, he \emph{is} his past self and must behave accordingly. Similarly, since his notes are colocated with his notepad when he returns from Rome, they must also have been present before he set off.

\subsection{The Entropy of a CTC}
This second interpretation severely challenges certain key assumptions of everyday computer science. For simplity let us assume that one program statement can be executed every second, and that it takes precisely $n$ seconds to traverse some given CTC exactly once. If we run a program on a computer following this CTC, then once the $n^\mathrm{th}$ statement has been executed, the entire system will have returned to its original spacetime location, and so must have returned to its original state. It follows that no irreversible process can be implemented on a computer following a CTC. And yet there is no obvious reason we shouldn't be able to load our computer with any program we like.

To avoid the apparent contradictions inherent in this situation, we need to re-appraise the nature of CTC computation. Since reversibility requires that no information is lost from the system, we have to conclude that when an irreversible procedure is executed on the computer, the information lost during program execution must be preserved somehow in the computer's environment, i.e. the CTC itself. It is well-known that black holes have an associated entropy  \cite{Bek73}. What we are suggesting here is that, given this second interpretation of CTC traversal, a CTC can also have an associated entropy; indeed CTCs are perfect information repositories in the sense that information cannot be lost (since it must be available to re-initialise the computation). It can be argued therefore that the CTC is an active component in the system which overrides the intended behaviour of the program, since the computer is forced to return to its initial state regardless of its underlying specification. Indeed, this can be seen as a mechanism enforcing the Novikov self-consistency principle \cite{Nov90} in the context of CTC computation.\footnote{Alternatively, one can argue that the information capacity of a CTC is strictly limited; this ensures that the computer cannot be provided with arbitrarily complex programs, and those programs cannot be supplied with arbitrarily complex inputs. Our research into this question is ongoing.}

\subsection{Single-traversable CTCs and computational speed-up}
Given our original interpretation of CTC traversal, \Ribbon shows that CTCs need not be traversable more than once. This second interpretation of CTC traversal likewise concludes that repeated traversal of a CTC cannot lead to computational speed-up, because any lengthy computation would be forced to return to its initialisation state rather than running to completion. The question remains whether it is possible to use CTCs to speed up computation, where we \emph{voluntarily} restrict ourselves to traversing CTCs no more than once. Indeed such schemes are described in the literature \cite{Bru03}, but these schemes make the additional
assumption of \emph{causal consistency}, using it to deduce that CTC-computation can solve PSPACE-problems \cite{AW09}.

As with CTCs themselves, there is no convincing evidence that causal consistency is experimentally necessary. What, then, can we deduce if we impose no additional constraints, and simply regard CTCs to be used as an implementation of time travel. As we illustrated in Fig. \ref{fig:algorithm-A'}, the availability of time travel can be used to show that P = NP. But the situation is not entirely clearcut, because we have assumed in Fig. \ref{fig:algorithm-A'} that information can be sent back to time $0$ no matter how long we have to wait for $A$ to complete its execution. But there is no guarantee that this is the case. For example, suppose the maximum time any CTC can take a traveller back is 5 seconds; then as soon as an input is provided which causes $A$ to run for more than 5 seconds, algorithm $A'$ will be invalid.

This suggests that, even in the presence of time travel, we cannot necessarily reduce problems in NP$\setminus$P (if any) to problems in P. For the remainder of this paper we therefore assume, to the contrary, that \PNP and ask what consequences this assumption entails.

\subsubsection{CTCs and \PNP.}
We assume the existence of some arbitrary observer (typically a computer) $O$. Given any spacetime location $X$ on $O$'s worldline, write $X^{+}$ for the set of timelike paths starting at $X$ that are traversable (in theory) by observers co-moving with $O$ at $X$. From $O$'s viewpoint, these paths are all future-pointing, and it is possible for $O$ to send information (e.g. by rocket) along any of these paths without requiring lightspeed or faster-than-light travel.

Some of these paths may intersect $O$'s worldline at points other than $X$. If these points lie to the past of $X$ from $O$'s point of view, then CTCs are present, and we can ask to what extent they can be exploited computationally. For simplicity, we will assume that $O$ can identify whether any given path in $X^+$ leads to a point $Y$ in his past, and can also identify the point $Y$ itself (i.e., how far back into his past the path takes him). It is extremely unlikely, of course, that such properties of CTCs would ever be so conveniently decidable.

Write \Past{X} for the set of all such points $Y$, i.e. those points on $O$'s past worldline that he can revisit by following paths in $X^{+}$. For each such intersection $Y$, write \HowFarBack{X}{Y} for the amount of time that originally passed, from $O$'s point of view, in travelling from $Y$ to $X$. In other words, if $O$ chooses to follow the path in question, how far into his own past will it take him? For simplicity we shall assume that all durations are measured in seconds.

Given any time $t$, write $X(t)$ for the point $X$ on $O$'s worldline that has time coordinate $t$ as coordinatized by $O$, and define the set of \emph{time differences} available at time $t$ to be the set
\[
   D(t) =
   \{
   		\HowFarBack{X(t)}{Y} : Y \in \Past{X(t)}
   \}
\]
and let
\[
   D^*(t) = \begin{cases}
   		\sup D(t) & \mbox{ if  $D(t)$ is bounded above } \\
      \mbox{undefined} & \mbox{ otherwise } \enspace .
      \end{cases}
\]
In essence, the function $D^*(t)$ tells us how far back $O$ can travel into his own past if he sets off on his journey at time $t$. If $D^*(t)$ is undefined, there is no limit to how far back $O$ can travel. Likewise we write $R(t)$ for the set of past times $O$ on his own worldline that are \emph{reachable} by setting off from $X(t)$, and $R^*(t)$ for the infimum of these reachable times, i.e.
\[
  R(t) = \{ t - t' : t' \in D(t) \}
\]
and
\[
	R^*(t) = \begin{cases}
   		\inf R(t) & \mbox{ if  $R(t)$ is bounded below } \\
      \mbox{undefined} & \mbox{ otherwise } \enspace .
      \end{cases}
\]

We show that when \PNP, $R^*(t)$ (equivalently, $D^*(t)$) must be defined for all sufficiently large $t$.

\begin{lemma}
\label{lem:R-is-downwards-closed}
If $t' \leq t$, then $R(t) \subseteq R(t')$.
\end{lemma}
\begin{proof}
Any path in $X(t)^+$ can be prepended by the section of $O$'s worldline running between times $t'$ and $t$ to generate a path in $X(t')^+$.
\qed
\end{proof}

\begin{theorem}
\label{thm:F-bounded-above}
Suppose \PNP. Then, for all sufficiently large $t$, $R(t)$ is bounded below, and hence $R^*(t)$ is defined.
\end{theorem}

\begin{proof}
Suppose to the contrary that there exists an unbounded increasing sequence of times $t$ at which $R(t)$ has no lower bound. By Lemma \ref{lem:R-is-downwards-closed}, $R(t)$ must be unbounded below for all $t$.

Let $A$ be a deterministic algorithm for solving some problem in NP. We use $A$ to define a new algorithm $B$ to be implemented on a computer co-moving with $O$, with the following behaviour (essentially a generalisation of the algorithm presented in Sect. \ref{sec:introduction}).

Given $n$, $O$ resets his clock to 0, waits one second and then checks whether any output has yet been generated. One second later he starts running $A(n)$. After $A(n)$ has eventually halted at time $t$ (say), $B$ travels back to some $T < 0$ in $R(t)$, waits until he re-encounters time $0$, and then publishes the result in time for his earlier incarnation to observe it at time $1$. As before this implies that problems in NP can be solved deterministically in constant time, whence P = NP (contrary to assumption).

Therefore no unbounded increasing sequence of times $t$ exists at which $R(t)$ is unbounded below, and the result follows.
\qed
\end{proof}

\begin{corollary}
\label{cor:R-interval}
Suppose $R^*(t)$ is defined, and suppose $t' \in (R^*(t), t]$. Then $R^*(t')$ is also defined, and $R^*(t) = R^*(t')$.
\end{corollary}
\begin{proof}
By assumption, $t' > R^*(t)$, so since $R^*(t) = \inf R(t)$ there must exist $T \in R(t)$ satisfying $t' > T \geq R^*(t)$. Consequently, $O$ can travel back from $X(t)$ to arrive back on his past worldline at time $T$, then wait (if necessary) until time $t'$ before setting off on any path in $R(t')$. Thus any past time reachable by $O$ from $X(t')$ is also reachable from $X(t)$, whence $R(t') \subseteq R(t)$. Since $R(t)$ is bounded below, the same must be true of $R(t')$, whence $R^*(t')$ is defined, as claimed.

Lemma \ref{lem:R-is-downwards-closed} tells us conversely that $R(t) \subseteq R(t')$ (since $t' \leq t$), and combining the two inclusions gives $R(t) = R(t')$, whence the claim follows.
\qed
\end{proof}

What do these simple results tell us? Theorem \ref{thm:F-bounded-above} tells us that, even if an observer is able to travel arbitrarily far back in time when he is young, he will eventually lose that capability as he grows older, and his reach into the past will become finite. Corollary \ref{cor:R-interval} then tells us as he grows older, he loses access to more and more of his past. This computational assumption seems to be telling us something also about cosmological structure. The exact nature of the relationship depends on an analysis of how fast the function $R^*$ grows, and remains an open question.

\section*{Acknowledgements}
I would like to thank Gergely Sz\'ekely, whose comments have greatly streamlined my thinking during the development of the ideas presented here.

\bibliography{pc2011stannett}

\end{document}